\documentclass[11pt,letterpaper]{article}
\usepackage[margin=1in]{geometry}
\usepackage[T1]{fontenc}
\usepackage{lmodern}
\usepackage{amsmath,amssymb,amsfonts,amsthm}
\usepackage{mathtools}
\usepackage{microtype}
\usepackage{booktabs,array}
\usepackage[dvipsnames,table]{xcolor}
\usepackage{graphicx}
\usepackage{tikz}
\usetikzlibrary{arrows.meta,shapes.multipart}
\usepackage{enumitem}
\usepackage{xspace}
\usepackage{aliascnt}
\usepackage{natbib}
\bibpunct{[}{]}{,}{a}{,}{,}
\usepackage{url}
\usepackage{hyperref}
\usepackage[nameinlink,noabbrev]{cleveref}

\NewCommandCopy{\originalcitep}{\citep}
\NewCommandCopy{\originalcitet}{\citet}
\NewCommandCopy{\originalcitealp}{\citealp}
\RenewDocumentCommand{\citep}{s}{\originalcitep*}
\RenewDocumentCommand{\citet}{s}{\originalcitet*}
\RenewDocumentCommand{\citealp}{s}{\originalcitealp*}
\hypersetup{
  hypertexnames=false,
  colorlinks=true,
  pdftitle={Nash Equilibria in Network Coordination Games are CLS-Complete},
  pdfauthor={},
  linkcolor={rgb:red,0.12;green,0.10;blue,0.48},
  citecolor={rgb:red,0.18;green,0.42;blue,0.18},
  urlcolor={rgb:red,0.05;green,0.25;blue,0.65}
}
\allowdisplaybreaks
\setlist[itemize]{leftmargin=2em,itemsep=2pt,topsep=4pt}
\setlist[enumerate]{leftmargin=2em,itemsep=2pt,topsep=4pt}
\theoremstyle{plain}
\newtheorem{maintheorem}{Main Theorem}
\newtheorem{theorem}{Theorem}[section]
\newaliascnt{lemma}{theorem}
\newtheorem{lemma}[lemma]{Lemma}
\aliascntresetthe{lemma}
\newaliascnt{proposition}{theorem}

\aliascntresetthe{proposition}
\newaliascnt{corollary}{theorem}

\aliascntresetthe{corollary}
\newaliascnt{claim}{theorem}

\aliascntresetthe{claim}
\theoremstyle{definition}
\newaliascnt{definition}{theorem}
\newtheorem{definition}[definition]{Definition}
\aliascntresetthe{definition}

\theoremstyle{remark}
\newaliascnt{remark}{theorem}

\aliascntresetthe{remark}
\newaliascnt{observation}{theorem}

\aliascntresetthe{observation}
\crefname{maintheorem}{Main Theorem}{Main Theorems}
\Crefname{maintheorem}{Main Theorem}{Main Theorems}
\crefname{theorem}{Theorem}{Theorems}
\crefname{lemma}{Lemma}{Lemmas}
\crefname{proposition}{Proposition}{Propositions}
\crefname{corollary}{Corollary}{Corollaries}
\crefname{claim}{Claim}{Claims}
\crefname{definition}{Definition}{Definitions}
\crefname{remark}{Remark}{Remarks}
\crefname{observation}{Observation}{Observations}
\newcommand{\CLS}{\ensuremath{\mathsf{CLS}}\xspace}
\newcommand{\PPAD}{\ensuremath{\mathsf{PPAD}}\xspace}
\newcommand{\PLS}{\ensuremath{\mathsf{PLS}}\xspace}
\newcommand{\TFNP}{\ensuremath{\mathsf{TFNP}}\xspace}
\newcommand{\QKKT}{\textnormal{\textsc{Quadratic-KKT}}\xspace}
\newcommand{\SFKKT}{\textnormal{\textsc{Bilinear-KKT}}\xspace}
\newcommand{\MaxCut}{\textnormal{\textsc{Local-MaxCut}}\xspace}
\newcommand{\NCN}{\textnormal{\textsc{Network-Coordination-Nash}}\xspace}
\newcommand{\NCGN}{\textnormal{\textsc{Network-Congestion-Nash}}\xspace}

\newcommand{\one}{\boldsymbol{1}}

\newcommand{\R}{\mathbb{R}}
\newcommand{\Q}{\mathbb{Q}}

\newcommand{\transpose}{\mathsf{T}}

\newcommand{\red}{\leq_p}
\newcommand{\shiftedbit}{\tilde{b}}
\newcommand{\vshiftedbit}{\tilde{\boldsymbol{b}}}

\usepackage{amsmath,amsfonts,bm}
\usepackage{multirow,array}
\usepackage{diagbox}

\newcommand{\defeq}{\coloneqq}

\def\va{{\boldsymbol{a}}}
\def\vb{{\boldsymbol{b}}}
\def\vc{{\boldsymbol{c}}}
\def\vd{{\boldsymbol{d}}}
\def\ve{{\boldsymbol{e}}}

\def\vj{{\boldsymbol{j}}}

\def\vs{{\boldsymbol{s}}}

\def\vu{{\boldsymbol{u}}}

\def\vx{{\boldsymbol{x}}}
\def\vy{{\boldsymbol{y}}}
\def\vz{{\boldsymbol{z}}}

\def\vF{{\boldsymbol{F}}}

\DeclareMathAlphabet{\mathsfit}{\encodingdefault}{\sfdefault}{m}{sl}
\SetMathAlphabet{\mathsfit}{bold}{\encodingdefault}{\sfdefault}{bx}{n}

\newcommand{\calG}{\ensuremath{\mathcal{G}}}

\newcommand{\calP}{\ensuremath{\mathcal{P}}}

\renewcommand{\bar}[1]{\overline{#1}}

\newcommand{\vbeta}{\boldsymbol{\beta} }

\newcommand{\mat}[1]{\mathbf{#1}}

\usepackage{amsthm}
\usepackage{mathtools}
\usepackage{amsmath}
\usepackage{bbm}
\usepackage{amsfonts}

\title{The Complexity of Nash Equilibrium in Network\\ Congestion and Coordination Games}
\author{
    Ioannis Anagnostides\\
    {\small Carnegie Mellon University}
    \and
    Ioannis Panageas\\
    {\small University of California, Irvine}
    \and
    Jingming Yan\\
    {\small University of California, Irvine}
}
\date{}
\begin{document}
\maketitle

\begin{abstract}
We show that computing a Nash equilibrium is $\CLS$-complete for linear network congestion and network coordination games. As a result, finding a KKT point of a bilinear polynomial is $\CLS$-complete. %Moreover, we prove that finding a Nash equilibrium in network congestion games is also $\CLS$-completem even under linear latencies.
%This resolves an open question asked by Cai and Daskalakis [SODA'11], Daskalakis and Papadimitriou [SODA'11], and more recently highlighted by~Fearnley, Goldberg, Hollender, and Savani [JACM'25]. Moreover, we prove that finding a Nash equilibrium in linear network congestion games is also $\CLS$-complete.
\end{abstract}

\section{Introduction}
\label{sec:introduction}

The search for Nash equilibria has been an enduring challenge at the heart of algorithmic game theory. Landmark results characterized its complexity in general-sum games~\citep{Daskalakis09:Complexity,Chen09:Settling,Rubinstein16:Settling}, crystallizing a beautiful theory revolving around \PPAD, a complexity class introduced by~\citet{Papadimitriou94:Complexity}. These results have opened a new research frontier, shifting the focus to understanding the complexity of Nash equilibria in more structured strategic interactions.

Few classes of games have received as much scrutiny in algorithmic game theory as \emph{network congestion games}, having played a central role in the development of the \emph{price of anarchy}~\citep{Koutsoupias99:Worst,Roughgarden02:Bad,Christodoulou05:Price}. Such games model traffic routing, where each player selects a path through a network to minimize their total travel time, and the latency on each edge depends on the number of players using it. Network congestion games admit a potential function and thus a \emph{pure} Nash equilibrium~\citep{Rosenthal73:Class,Monderer96:Potential}. \citet{Fabrikant04:Complexity} famously proved that finding one is complete for $\PLS$---a complexity class introduced by~\citet{Johnson88:Easy}; this result prompted extensive work on its inapproximability~\citep{Skopalik08:Inapproximability,Caragiannis11:Efficient}.

Another well-studied class is \emph{network coordination games}~\citep{DP2011}. Such a game is represented by a graph whose nodes correspond to players and whose edges specify two-player games with \emph{identical interests}. Each player's utility is the sum of their payoffs from incident edges. This can be equivalently thought of as a \emph{polymatrix game}~\citep{Cai16:Zero} in which each edge corresponds to a coordination game. Just like network congestion games, network coordination games admit a \emph{pure} Nash equilibrium, and \citet{CD2011} showed that finding one is again \PLS-complete, making the problem computationally equivalent to that in network congestion games. However, the complexity of finding a (\emph{mixed}) Nash equilibrium in either network congestion or network coordination games has remained an elusive open problem.

To cast light on the complexity of problems that belong to both $\PPAD$ and $\PLS$, \citet{DP2011} introduced in their seminal work the complexity class $\CLS$, which stands for \emph{continuous local search}. They placed the problem of finding a Nash equilibrium in network coordination and network congestion games in $\CLS$, but they left open its precise complexity. Following the breakthrough result of~\citet{FGHS2022} that established $\CLS = \PPAD \cap \PLS$, there has been much progress shaping the landscape within $\CLS$. $\CLS$-hardness provides strong evidence of intractability. Notably, \citet{HY20:Hardness} established black-box lower bounds and white-box hardness under cryptographic assumptions for $\CLS$.

A key reference point to our work is the $\CLS$-completeness of~\citet{BR2021} for finding mixed Nash equilibria in congestion games. Their reduction relies on identical-interest interactions involving \emph{five} players, whereas network \emph{coordination} games permit only \emph{pairwise} (two-player) interactions. Moreover, \emph{network} congestion games require each player’s strategies to be paths in a common network, so their reduction did not establish hardness for that more structured class; this open problem was also highlighted by~\citet{FGHS2022}. More recently, \citet{FGHS2025} proved that finding a \emph{Karush–Kuhn–Tucker (KKT)} point---an optimization analogue of a Nash equilibrium---of a quadratic polynomial over a box is $\CLS$-complete. Their construction, however, crucially hinges on squared variable terms, which cannot arise in the bilinear potential of a network coordination game. The complexity of finding a KKT point of a bilinear polynomial was the main question posed by~\citet{FGHS2025}, and is arguably the main piece missing from characterizing the complexity of KKT points.

\subsection{Our results}
\label{sec:results}

In more detail, the problem of finding a Nash equilibrium of a network congestion game can be written as follows.

\medskip
\noindent\NCGN
\begin{itemize}
    \item \emph{Input:} A directed graph $G=(V,E)$ and $N$ players, where every player $i\in[N]$ has a source $s_i\in V$ and a
    destination $t_i\in V$ connected by a directed path. Every edge $e\in E$
    has a latency function $\ell_e(k)=\alpha_e k$, where
    $\alpha_e\in\Q_{\geq0}$ and $k$ is the number of players using $e$.
    Each player $i$ chooses a simple $s_i$-$t_i$ path and incurs the sum of the latencies of its edges.
    \item \emph{Output:} An exact Nash equilibrium, represented for each
    player $i$ by a list of polynomially many simple $s_i$-$t_i$ paths
    and the rational probabilities with which the player chooses them.
\end{itemize}

The complexity of this problem has remained open since at least the work of~\citet{Daskalakis06:Game}. Our first main result resolves this question.

\begin{maintheorem}
\label{thm:congestion-hardness}
\NCGN is $\CLS$-complete.
\end{maintheorem}
Network congestion games are often defined more broadly, with nonlinear latencies (\emph{e.g.},~\citealp{Roughgarden15:Intrinsic}). Our result shows that \emph{even with linear latencies}, the problem is $\CLS$-complete; prior to our work, this had remained open even under more general families of latencies.

We next turn our attention to network coordination games. For our hardness result, it will suffice to restrict to binary-action games, formally described below.

\medskip
\noindent\NCN
\begin{itemize}
    \item \emph{Input:} An undirected graph $G=(V,E)$, where every player
    $i\in V$ has two actions $\{0,1\}$, and every edge $(i,j)\in E$
    is associated with a rational $2\times 2$ common-payoff matrix with
    entries in $[0,1]$. Each player's payoff is the sum of the payoffs
    received from its incident edges.
    \item \emph{Output:} An exact Nash equilibrium, represented by a vector
    $\vx\in[0,1]^{|V|}$, where $x_i$ is the probability with which player $i$
    chooses action $1$.
\end{itemize}

The complexity of this problem was left open by~\citet{CD2011} and~\citet{DP2011}, and was more recently highlighted by~\citet{FGHS2025}. Our second contribution is to resolve this problem.

\begin{maintheorem}
\label{thm:main}
\NCN is $\CLS$-complete.
\end{maintheorem}

An immediate consequence of this, combining with the membership results of~\citet{Anagnostides26:Complexity}, is $\CLS$-completeness for \emph{perfect} and \emph{proper} equilibria---two central equilibrium refinements introduced by~\citet{Selten75:Reexamination} and~\citet{Myerson78:Refinements}, respectively---in network coordination games.

Another important implication of our results is that computing a KKT point of a \emph{bilinear} polynomial---a quadratic polynomial in which no monomial contains the square of a variable---is \CLS-complete.

\begin{maintheorem}
\label{thm:squarefree}
Computing a KKT point of a bilinear polynomial over $[0,1]^m$ is $\CLS$-complete.
\end{maintheorem}
This resolves the main lingering open question concerning the complexity of computing KKT points in polynomial optimization. Beyond their intrinsic interest, KKT points play a key role in many other problems, as discussed further in~\Cref{sec:related-work}. Our result shows that existing hardness results persist even under substantially stronger structural restrictions.

\subsection{Technical overview}
\label{sec:overview}

We now present an overview of our reduction. The starting point is the $\CLS$-complete problem \QKKT \citep{FGHS2025}. To reduce \QKKT\ to \SFKKT, the main technical challenge is to eliminate all square terms while ensuring that every KKT point of the constructed instance can be decoded into a source solution. A natural starting point is to introduce two copies $\vx,\vy$ and replace each square term by the product of its two copies. If the two copies agree at a KKT point, that is, $\vx=\vy$, their common value can be used to recover a solution
to the original \QKKT instance. However, the crux lies in ensuring that two copies will always agree at any KKT point.

More fundamentally, every bilinear potential over $[0,1]^n$ admits a global minimum---and hence a KKT point---at a Boolean point in $\{0, 1\}^n$. %Indeed, since a bilinear function is affine in each coordinate separately, starting from any global minimum we can move each coordinate to $0$ or $1$ while preserving the minimum objective value. 
By contrast, a quadratic objective need not admit any Boolean KKT point. Thus, not every KKT point of the constructed instance can be decoded directly into a KKT point of the source instance. Instead, we need a different way to interpret the Boolean solutions.

\paragraph{Basic idea.} The core of our reduction is an \textit{either} argument, which is a two-way decoding technique. This was used by~\citet{FGHS2022} in their breakthrough result, and more recently by~\citet{BC2026} to settle the complexity of computing first-order stationary points in min-max optimization. The overview of our reduction is given in~\Cref{fig:reduction-overview}.

In a nutshell, given a \QKKT\ instance $p$ over $[0,1]^m$, our construction provides two ways to recover a KKT point of $p$. Alongside the direct two-copy simulation of $p$, we construct a Boolean local search instance with potential $h$, whose local minima can be decoded into KKT points of $p$. We then construct a bilinear potential $\Phi$ that incorporates both the two-copy representation of $p$ and the Boolean potential $h$. For each coordinate $q$, we introduce a pair of \emph{control variables}
together with an associated pair of \emph{copy variables} $(\vx_q,\vy_q)$. These controls determine which of two decoding mechanisms applies.
If the control value is fractional for some coordinate $q$, a KKT point of $p$ can be recovered directly. Otherwise, all control values are Boolean and together form a local minimum of $h$, which can again be decoded into a KKT point of $p$. This dichotomy ensures that every KKT point of $\Phi$ can be mapped to a KKT point of the original instance $p$. Finally, we translate the resulting bilinear potential into a network coordination game and extend the hardness result to linear network congestion games through a separate reduction. %The two decoding scenarios and the reduction to linear network congestion games are summarized in~\Cref{fig:reduction-overview}.

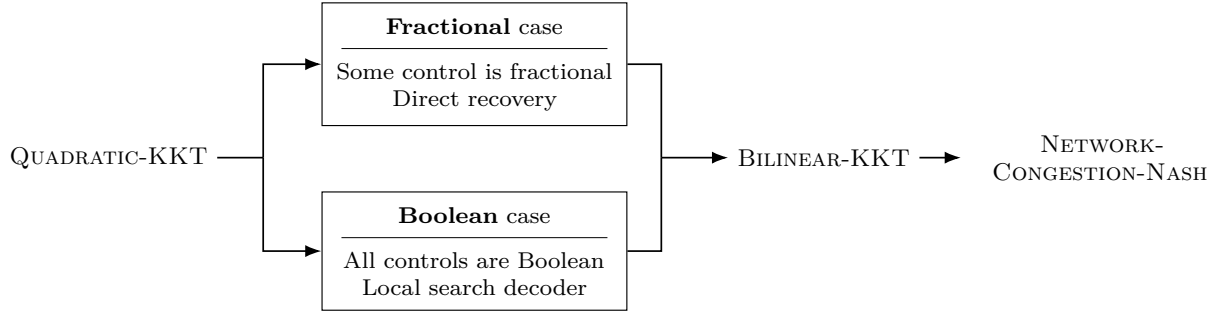
\begin{figure}[htbp]
\centering
\resizebox{\linewidth}{!}{%
\begin{tikzpicture}[
  problem/.style={align=center, font=\fontsize{8.5}{10}\selectfont,
    inner sep=3pt, minimum height=.75cm},
  case box/.style={draw, rectangle split, rectangle split parts=2,
    rectangle split draw splits=false, align=center,
    font=\scriptsize, inner sep=5pt, text width=3.4cm},
  flow/.style={-{Latex[length=2mm]}, semithick}
]
\node[problem] (source) at (0,0) {\QKKT};
\node[case box] (fractional) at (4.5,1.15)
  {\textbf{Fractional} case\nodepart{second}Some control is fractional\\Direct recovery};
\node[case box] (boolean) at (4.5,-1.15)
  {\textbf{Boolean} case\nodepart{second}All controls are Boolean\\Local search decoder};
\foreach \case in {fractional,boolean} {
  \draw[line width=.2pt,shorten <=7pt,shorten >=7pt]
    (\case.text split west) -- (\case.text split east);
}
\node[problem] (target) at (8.8,0) {\SFKKT};
\node[problem, text width=3.2cm] (congestion) at (12.2,0) {\NCGN};
\draw[semithick] (source.east) -- (1.9,0);
\draw[flow] (1.9,0) |- (fractional.west);
\draw[flow] (1.9,0) |- (boolean.west);
\draw[semithick] (fractional.east) -| (6.8,0);
\draw[semithick] (boolean.east) -| (6.8,0);
\draw[flow] (6.8,0) -- (target.west);
\draw[flow] (target.east) -- (congestion.west);
\end{tikzpicture}
}
\caption{The two decoding scenarios in the reduction from \QKKT\ to \SFKKT.
Every KKT point of the constructed bilinear potential yields a KKT point of the source instance. We finally reduce network coordination games to linear congestion games.}
\label{fig:reduction-overview}
\end{figure}

\paragraph{The two-copy simulation.}
In more detail, we start from \QKKT, the problem of finding KKT points of
\[
\min_{\vz\in[0,1]^m}
p(\vz)
=
\frac{1}{2}\vz^\mathsf{T} \mat{Q}\vz+\vc^\mathsf{T} \vz,
\]
where $\mat{Q}$ is symmetric and $|\mat{Q}_{ii}|\leq 1$ for every $i\in[m]$. To eliminate the square terms, we introduce two copies $\vx,\vy$ and replace each square term by the product of the corresponding copy variables. This yields a bilinear polynomial $P$ satisfying
$
P(\vx,\vy)=2p(\vx)
$
whenever $\vx = \vy$. However, simply replacing square terms with products does not by itself guarantee that the two copies agree at a KKT point. Indeed, writing $d_i=x_i-y_i$, we have
$
\partial_{x_i}P(\vx,\vy)-\partial_{y_i}P(\vx,\vy)
=
-\mat{Q}_{ii}d_i.
$
When $\mat{Q}_{ii}>0$, this derivative difference drives the two copies apart. We therefore add the bilinear term
$
-\sum_{i<j}(x_i-y_i)(x_j-y_j)
$
to obtain $\widehat P$. The resulting derivative difference takes the form
$
\partial_{x_i}\widehat P-\partial_{y_i}\widehat P
=
(2-\mat{Q}_{ii})d_i-2\sum_j d_j.
$
Since $2-\mat{Q}_{ii}>0$, the first term now has the desired sign for ensuring agreement between the two copies once the remaining term vanishes. This is precisely the situation in what we call the \emph{fractional case}, where the other components of $\Phi$ cancel the term $-2\sum_j d_j$.

\paragraph{From continuous to discrete local search.}
In~\Cref{sec:boolean-alternative}, we construct a discrete local search potential $h$ that will be used to decode the Boolean case---when all control variables are either $0$ or $1$. To this end, we discretize the domain $[0,1]^m$ using a sufficiently fine grid that can be represented with polynomially many bits. We define a local move by increasing or decreasing a single coordinate to the adjacent grid value, provided the resulting point remains feasible. At any local minimum, no such move can decrease $p$, which implies approximate KKT conditions for the original continuous problem. Choosing the grid sufficiently fine, we can recover an exact KKT point of $p$ through the rounding argument of~\citet{FGHS2025}. Finally, applying the standard $\PLS$ reduction of \citet{SY1991}, we obtain a bilinear potential $h$ such that every local minimum of $h$ can be decoded into a KKT point of $p$.

\paragraph{Constructing the bilinear potential $\Phi$.} Next, in \Cref{sec:reduction}, we combine the two-copy simulation $\widehat P$ and the Boolean local search potential $h$ into a single bilinear potential $\Phi$. The key challenge is to couple the two constructions so that the KKT conditions with respect to the control variables enforce the desired behavior in both cases: when some control variable is fractional, we eliminate the error in the two-copy simulation, while when the control variables are Boolean, we get local optimality of $h$.

As discussed, the two-copy simulation $\widehat P$ leaves the undesired term $-2\sum_j d_j$ in the difference between the gradients of copy variables.
On the other hand, local optimality of the Boolean potential $h$ is encoded by its gradient $\nabla h$. Our goal is therefore to bring these
two quantities into the same gradient expressions. We achieve this by shifting each control value by the disagreement between the corresponding copy variables and evaluating $h$ at the resulting vector
$\widetilde{\vb}$. In this way, the local search potential depends directly on the copy disagreement.

Putting these ingredients together, we define $\Phi$ by combining the Boolean local search potential $h(\widetilde{\vb})$, the two-copy simulation $\widehat P$, and additional bilinear interaction terms designed to align their gradient contributions. With this choice, the undesired term from the two-copy simulation and the contribution from the Boolean potential are absorbed into a single residual term $R_q$. In the fractional case, the KKT conditions of the controls force $R_q$ to vanish, which ultimately drives $\vx_q=\vy_q$. In the Boolean case, the sign of the residual $R_q$ instead yields the local optimality condition of $h$. This is the fundamental building block behind our \emph{either} construction. We provide the details in~\Cref{sec:reduction}.

\paragraph{Reducing network coordination to network congestion games.} Finally, to establish~\Cref{thm:congestion-hardness}, we give a reduction from network coordination games to network congestion games, even when the latter has linear latency functions. To do so, we adapt the construction of~\citet{Ackermann08:Impact}, which reduced \emph{pure} Nash equilibria in \emph{quadratic threshold games}---a special class of congestion games---to pure Nash equilibria in linear network congestion games. The main difference in our reduction is that we start from a network coordination game and we also need to account for mixed Nash equilibria. Nonetheless, our construction is a relatively straightforward adaptation; we are surprised that, to our knowledge, this reduction was not observed before. We provide the details in~\Cref{sec:network-congestion}.

\paragraph{The \emph{either} argument and broader picture.} Taken together, our main technical contribution is a new two-way decoding argument, which goes beyond the scope of Nash equilibria in network congestion and network coordination games. More broadly, we view this argument as part of recent complexity progress in $\TFNP$~\citep{FGHS2022,BC2026}. We hope that this technique can facilitate further progress in
understanding the complexity landscape of $\TFNP$.

\subsection{Further related work}
\label{sec:related-work}

Several equilibrium computation problems have been shown to be $\CLS$-complete. \citet{Ghosh24:Complexity} and~\citet{Tewolde25:Computing} established $\CLS$-completeness for finding a symmetric Nash equilibrium of a symmetric identical-interest game; the former result holds even with two players. Moreover, certain solution concepts in imperfect-recall games are also $\CLS$-complete~\citep{Tewolde23:Computational,Tewolde24:Imperfect}. Combining our main result for network coordination games with the reduction of~\citet{Tewolde23:Computational}, it follows that computing an exact \emph{CDT equilibrium} in a single-player imperfect-recall game is $\CLS$-complete even with no absentmindedness, two actions per information set, and tree depth three (including an initial chance move). In \emph{adversarial team games}, where a team of identical-interest players competes against a single adversary, \citet{Anagnostides26:Algorithms} established $\CLS$-completeness; the hardness follows readily from~\citet{BR2021} in the multi-player setting. \citet{Hollender25:Complexity} proved $\CLS$-hardness for two-team \emph{polymatrix} games. Our main result shows that $\CLS$-hardness persists \emph{even without an adversary}. \citet{Anagnostides25:Complexity} later showed $\CLS$-hardness even in three-player polymatrix adversarial team games. 

Beyond equilibrium computation, $\CLS$ plays a key role in the complexity of fixed-point problems involving contractions under general metrics and metametrics~\citep{Fearnley17:CLS,Daskalakis18:Converse,Ishizuka21:Complexity}.

There are still many problems that are known to be in $\CLS$, but whose complexity remains open. For example, \citet{Fearnley19:Unique} introduced \emph{unique end of potential line}, a subclass of $\CLS$ that captures several problems with unique solutions. This was inspired by the class \emph{end of potential line}~\citep{Fearnley17:CLS}, which was subsequently shown to coincide with $\CLS$~\citep{Goos24:Further}. It would be interesting to see if the technical approach we develop in this paper can be leveraged to make progress in such problems.

\section{Preliminaries}
\label{sec:preliminaries}
\paragraph{Notation.} For a positive integer $d$, we use the notation $[d]\defeq \{1,\ldots,d\}$. We write $\ve_i$ for the $i$-th standard basis vector and $\one$ for the all-one vector, whose dimension will be clear from the context. All polynomial coefficients and payoff entries are rational and are given explicitly as part of the input. A polynomial time reduction between search problems consists of a polynomial time instance map and a polynomial time decoder that maps every
accepted target solution to an accepted source solution. We write $\red$ to denote such a reduction.

\subsection{Quadratic KKT conditions}
\label{sec:kkt}

For a vector $\vz \in [0, 1]^m,$ we define the quadratic function
\begin{equation}
 p(\vz)=\frac12\vz^\transpose \mat{Q} \vz+\vc^\transpose \vz,
\label{eq:source}
\end{equation} where $\mat{Q}$ is a symmetric matrix with $\mat{Q} = \mat{Q}^\transpose\in\Q^{m\times m}$.
We let $\vF(\vz) \defeq \nabla_{\vz} p(\vz) = \mat{Q} \vz+\vc$.

\begin{definition}[KKT points]
\label{def:kkt}
A point $\vz\in[0,1]^m$ is a \emph{KKT point} of $p$ if, for every $i\in[m]$,
\begin{equation}
 z_i>0\ \Longrightarrow\ F_i(\vz)\leq 0 \quad
 \text{and} \quad
 z_i<1\ \Longrightarrow\ F_i(\vz)\geq0.
 \label{eq:kkt} \tag{KKT}
\end{equation}
In particular, any interior coordinate has zero derivative at an exact KKT point. At
$0$ the derivative must be nonnegative and at $1$ nonpositive. Since positive rescaling preserves exact KKT points, for the rest of the paper, we assume
$
 |\mat{Q}_{ii}|\leq1 \textrm{ for }i\in[m].
$
For some $\epsilon\geq0$, a point $\vz\in[0,1]^m$ is an
\emph{$\epsilon$-KKT point} if, for every $i\in[m]$,
\[
z_i>0\ \Longrightarrow\ F_i(\vz) \leq \epsilon
\quad\text{and}\quad
z_i<1\ \Longrightarrow\ F_i(\vz) \geq - \epsilon.
\]
\end{definition}

\QKKT asks for a KKT point of \eqref{eq:source}, given explicitly as part of the input. \citet{FGHS2025} established that \QKKT is $\CLS$-complete. Its bilinear restriction, which we denote by \SFKKT, additionally requires $\mat{Q}_{ii}=0$ for every $i \in [m]$, so that the underlying function contains no square variable terms.

We will use the following rounding lemma~\citep[Lemma A.1]{FGHS2025}, which in turn relies on the rounding argument of~\citet{Etessami10:Complexity}.

\begin{lemma}[\citealp{FGHS2025}]\label{lem:exact-recovery}
    Given an instance p of \QKKT, we can compute in polynomial time an $\epsilon> 0$ such that from any $\epsilon$-KKT point of $p$, we can obtain an exact KKT point of $p$ in polynomial time.
\end{lemma}

\subsection{Network coordination games}
\label{sec:game-def}

A binary-action \emph{network coordination game} $\calG$ is given by an undirected graph
$G=(V,E)$, with $|V|=N$, and a $2\times 2$ matrix
$\mat{A}^{ij}$ for every edge
$(i,j)\in E$ with $i<j$. At a pure strategy profile $\va = (a_1,\ldots,a_N)$, both player $i$ and player $j$ receive payoff $\mat{A}^{ij}_{a_ia_j}$ from that edge. The total utility of player $i$ is the sum of its incident edge payoffs.

Let $x_i\in[0,1]$ be player $i$'s probability to play action $1$. For a mixed profile $\vx = (x_1, \dots, x_N)$, we write $u_i(\vx)$ for the expected payoff. Our focus is on (mixed) Nash equilibria~\citep{Nash50:Non}.

\begin{definition}
A strategy profile $\vx \in [0,1]^N$ is a \emph{Nash equilibrium} of a binary-action network coordination game if for every player $i \in [N]$ and action $a_i' \in \{0, 1\}$, $u_i(\vx) \geq u_i(a_i', \vx_{-i})$.
\end{definition}

Given a joint (possibly mixed) strategy profile $\vx \in [0, 1]^N$, the bilinear potential is given by
\begin{equation*}
 \Phi_{\calG}(\vx)=\sum_{(i, j)\in E}\Big[
 (1-x_i)(1-x_j) \mat{A}^{ij}_{00}+(1-x_i)x_j \mat{A}^{ij}_{01}
 +x_i(1-x_j)\mat{A}^{ij}_{10}+x_ix_j \mat{A}^{ij}_{11}\Big].
 %\label{eq:game-potential}
\end{equation*}

We next point out a well-known equivalence between Nash equilibria in potential games and KKT points. For completeness, we include the simple proofs.

\begin{lemma}
\label{lem:potential-game}
Let $\calG$ be a binary-action network coordination game. A joint strategy profile $\vx\in[0,1]^N$ is a Nash equilibrium of $\calG$ if and only if it is a KKT point of the problem
\begin{align*}
    \min_{\vx\in[0,1]^N} -\Phi_{\calG}(\vx).
\end{align*}
\end{lemma}
\begin{proof}
We have $-\partial_i \Phi_\calG(\vx) = u_i(0,\vx_{-i}) -u_i(1,\vx_{-i})$. At a Nash equilibrium, player $i$ can assign positive probability to action $1$ only if $-\partial_i \Phi_\calG(\vx) \le 0$, and to action $0$ only if $-\partial_i \Phi_\calG(\vx) \ge 0$. Conversely, it follows that the KKT conditions are also sufficient to ensure the Nash equilibrium condition.
\end{proof}

\begin{lemma}
\label{lem:game-realization}
Consider a bilinear quadratic polynomial $p: [0, 1]^N \ni \vz \mapsto \vbeta^\transpose \vz + \sum_{i<j}w_{ij}z_iz_j$ with rational coefficients. There is a polynomial time construction of a
binary-action network coordination game with $N+1$ players and edge payoffs in
$[0,1]$ such that every Nash equilibrium of the constructed game induces a KKT point of $p$.
\end{lemma}

\begin{proof}
For each pair $i<j$ with $w_{ij}\neq 0$, we define the common payoff on edge $(i,j)$ to be
$-w_{ij}a_ia_j$ when the endpoints choose actions $a_i,a_j\in \{0,1\}$. To realize the linear terms $\vbeta^\top \vz$, we add one additional
player $0$. On edge $(i,0)$ the common payoff is $-\beta_i a_i$,
independent of $a_0$. For each original player $i \in [N],$ the payoff difference is $u_i (1, \vz_{-i}) - u_i(0, \vz_{-i}) = -\beta_i-\sum_{j\ne i}w_{ij}z_j=-\partial_i p(\vz)$, with $w_{ji} =w_{ij}$. The best-response conditions are therefore exactly the KKT conditions for minimizing $p$. Finally, rescaling the edge payoffs to $[0, 1]$ preserves all best responses. The constructed game has $N+1$ players, at most $O(N^2)$ edges, and payoff entries whose encoding lengths are polynomially bounded in the description of $p$.
\end{proof}

% END sections/03_games_and_membership.tex

\subsection{Network congestion games}
\label{sec:congestion-def}

A \emph{network congestion game} is given by a directed graph $G=(V,E)$, a set of $N$ players, and a latency function $\ell_e$ for every edge $e\in E$. Each player $i\in[N]$ has a source $s_i\in V$ and a destination $t_i\in V$. We write $\calP_i$ for the nonempty set of all simple $s_i$-$t_i$ paths, which is player $i$'s set of actions.

At a pure strategy profile $\boldsymbol{P}=(P_1,\ldots,P_N)$, where $P_i\in\calP_i$, the load on an edge is the number of players using it, that is,
\[
 k_e(\boldsymbol{P})\defeq \bigl|\{i\in[N]:e\in P_i\}\bigr|.
\]
Each player incurs as cost the sum of the latencies along its chosen path,
\[
 \mathrm{cost}_i(\boldsymbol{P})
 \defeq \sum_{e\in P_i}\ell_e\bigl(k_e(\boldsymbol{P})\bigr).
\]
(To stay consistent with existing work in network congestion games, unlike network coordination games, here we take a cost minimization perspective.) We consider nonnegative affine latencies $\ell_e(k)=\alpha_e k+\beta_e$, where
$\alpha_e,\beta_e\in\Q_{\geq0}$. A latency is \emph{linear} if $\beta_e=0$, so that $\ell_e(k)=\alpha_e k$.

A mixed strategy $\vx_i$ is a probability distribution over $\calP_i$. Under a strategy profile $\vx=(\vx_1,\ldots,\vx_N)$, players draw their paths independently. We write $\mathrm{cost}_i(\vx)$ for
player $i$'s expected cost, and $\mathrm{cost}_i(P_i',\vx_{-i})$ for its expected cost when it instead chooses the fixed path $P_i'\in\calP_i$.

\begin{definition}
\label{def:congestion-ne}
A joint strategy profile $\vx$ is a \emph{Nash equilibrium} of a network congestion game if for every player $i \in[N]$ and path $P_i'\in\calP_i$,
\[
 \mathrm{cost}_i(\vx)
 \leq \mathrm{cost}_i(P_i',\vx_{-i}).
\]
\end{definition}

% BEGIN sections/04_boolean_alternative.tex
\section{A discrete local search alternative}
\label{sec:boolean-alternative}
\label{sec:boolean-def}

Our main reduction relies on the observation that one can recover a solution to the \QKKT instance through a discrete local search problem defined by a Boolean objective $h$. In particular, we construct $h$ so that every local minimum can be decoded into an exact KKT point of the original quadratic objective $p$. %Later in the construction, we will show that when the control variables take Boolean values, they induce a local minimum of $h$.

For a Boolean vector $\vb\in\{0,1\}^n$, we let $\vb^{\oplus q}$ denote the vector obtained
by flipping bit $q$. We introduce a bilinear objective with integer coefficients
\begin{equation*}
 h(\vb)=\sum_{q=1}^n\beta_q b_q+\sum_{q<r}w_{qr}b_qb_r.
\end{equation*}
A Boolean local minimum satisfies
$h(\vb^{\oplus q})\geq h(\vb)$ for every $q$.
Similar to \eqref{eq:source}, we set
$w_{rq}=w_{qr}$ and $w_{qq}=0,$ and use the same polynomial to define $h(\vu)$ over the entire real domain $\R^n$. We write
\begin{equation*}
 \nabla h(\vu)_q\defeq \frac{\partial h}{\partial u_q}(\vu)
       =\beta_q+\sum_{r\ne q}w_{qr}u_r.
\end{equation*}
Since $h$ is affine in each coordinate separately,
\begin{equation}
 h(\vb^{\oplus q})-h(\vb)=(1-2b_q)\nabla h(\vb)_q.
 \label{eq:flip}
\end{equation}
Since all the coefficients in $h$ are integers, the derivative $\nabla h(\vu)_q$ is an integer at every Boolean assignment.

We now construct a discrete local search instance from the original quadratic objective.

\begin{lemma}
\label{lem:grid-pls}
Given a quadratic objective $p$ of a \QKKT instance over $[0, 1]^m$, one can construct in polynomial time a discrete local search instance such that every local optimum can be converted in polynomial time into an exact KKT point of $p$.
\end{lemma}
\begin{proof}
By \Cref{lem:exact-recovery}, we can compute a rational tolerance $\epsilon>0$ with polynomial encoding length such that any $\epsilon$-KKT point of $p$
can be converted into an exact KKT point.

Fix some tolerance $\epsilon>0$. For each coordinate, consider a uniform grid on $[0,1]$ with $M+1$ points and grid spacing $\delta=1/M$, where $M=2^k-1$ for some integer $k\ge 1$ chosen so that $\delta/2\le \epsilon$. Since $\epsilon$ has polynomial encoding length, one can choose $k=O(\log(1/\epsilon))$, which is polynomially bounded in the description length of $p$.

We then consider the following discrete local search problem over this grid. A feasible solution is represented by a bit string of length $mk$, partitioned into $m$ blocks of $k$ bits. The $i$-th block encodes an integer $j_i\in \{0,\ldots,M\}$, corresponding to the grid coordinate $z_i=j_i/M$. Thus, every feasible solution encodes a valid grid point
$
\vz=\left(\frac{j_1}{M},\ldots,\frac{j_m}{M}\right) \in [0, 1]^m.
$
We define two grid points to be neighbors if they only differ by $\pm\delta \ve_i$ for one coordinate $i\in[m]$. Each grid point has at most $2m$ neighbors.

For any grid point $\vz$, use $p(\vz)$ as the objective value we aim to minimize, so that a neighboring point $\vz'$ is improving whenever $p(\vz')<p(\vz)$. Now, let $\vz$ be a local optimum with respect to this neighborhood. Since
$
\vF(\vz) \defeq \nabla p(\vz)= \mat{Q}\vz+\vc,
$
if $\vz+\delta \ve_i$ is a feasible point, local optimality implies
\begin{align*}
    0
\le
p(\vz+\delta \ve_i)-p(\vz)
=
\delta F_i(\vz)+\frac{\delta^2}{2}\mat{Q}_{ii}.
\end{align*}
Similarly, if $\vz-\delta \ve_i$ is feasible, then
\begin{align*}
    0
\le
p(\vz-\delta \ve_i)-p(\vz)
=
-\delta F_i(\vz)+\frac{\delta^2}{2}\mat{Q}_{ii}.
\end{align*}
If $z_i>0$, the neighboring grid point $\vz-\delta\ve_i$ is feasible. Using local optimality and $|\mat{Q}_{ii}|\le 1$, we obtain
$F_i(\vz)\le \delta/2 \le \epsilon.$
Similarly, if $z_i < 1$, it holds that
$F_i(\vz)\ge -\delta/2\ge -\epsilon.$
Hence, $\vz$ is an $\epsilon$-KKT point of $p$, and \cref{lem:exact-recovery} completes the recovery of an exact KKT point.

Each grid point is represented using only $mk$ bits, and the $2m$ neighbors can be enumerated in polynomial time. Finally, to cast this discrete local search problem in the standard \PLS form, we rescale $p(\vz)$ to an equivalent nonnegative integer-valued objective. Let $D$ be a common positive denominator of the entries of $\mat{Q}$ and $\vc$. Writing $\vz=\vj/M$, we have
$2DM^2p(\vz) = D\vj^\transpose \mat{Q}\vj + 2DM\vc^\transpose\vj \in\mathbb Z$. Multiplying by $2DM^2>0$ and adding a sufficiently large integer constant preserves all improving moves while producing a nonnegative objective with integer values. Hence, our discrete local search problem admits a polynomial-size representation and polynomial time neighborhood evaluation, as required for \PLS.
\end{proof}

\begin{lemma}
\label{lem:boolean-alternative}
Given a rational quadratic objective $p$ over $[0,1]^m$, one can construct in polynomial time a bilinear polynomial $h$ over the discrete set $\{0,1\}^n$ with integer coefficients, together with a polynomial time decoder $\mathcal D_p$, such that
\begin{equation*}
 \vb\text{ is a local minimum of }h
 \quad\Longrightarrow\quad
 \mathcal D_p(\vb)\text{ is an exact KKT point of }p.
\end{equation*}
\end{lemma}
\begin{proof}
By \cref{lem:grid-pls}, given any \QKKT instance $p$, we can construct a discrete local search problem such that every local optimum can be decoded into an exact KKT point of $p$. We apply the standard PLS reduction of \citet{SY1991} to transform this discrete local search problem into an instance of weighted \MaxCut.

A cut can be represented by a Boolean vector $\vb\in\{0,1\}^n$. We define $h$ to be the negative cut weight of the resulting graph, that is,
\begin{equation*}
    h(\vb) = -\sum_{(q,r) \in E} w_{qr}
    \bigl(b_q+b_r-2b_qb_r\bigr).
\end{equation*}
$h$ is a bilinear polynomial with integer coefficients. Moreover, since flipping a single coordinate of $\vb$ corresponds to moving one vertex to the opposite side of the cut, $\vb$ is a local minimum of $h$ if and only if it represents a locally maximal cut.
By this reduction and \cref{lem:grid-pls}, every local minimum of $h$ can be converted in polynomial time into an exact KKT point of $p$, yielding the desired polynomial-time decoder $\mathcal D_p$.
\end{proof}

% BEGIN sections/05_reduction.tex
\section{The main reduction}
\label{sec:reduction}

We now give the reduction from \QKKT\ to \SFKKT. Starting from the quadratic objective $p$ in \eqref{eq:source}, normalized so that $|\mat{Q}_{ii}|\leq 1$ for every $i\in[m]$, let $h$ and $\mathcal D_p$ be obtained from \cref{lem:boolean-alternative}. We construct a bilinear potential $\Phi$ by introducing, for each Boolean coordinate $q$ of $h$, a bilinear representation of the source objective $p$ together with a pair of what we call \emph{control variables}. The construction is designed so that every KKT point of $\Phi$ can be decoded in one of two ways: either directly into a KKT point of $p$, or into a local minimum of $h$, which can then be mapped to a KKT point of $p$ through $\mathcal D_p$.

% Let $p$ be as in \eqref{eq:source}, normalized to satisfy
% $|\mat{Q}_{i,i}| \leq 1$ for all $i$. Let $h$ and $\mathcal D_p$ be given by
% \cref{lem:boolean-alternative}. The construction below uses only this
% lemma's decoding property, the quadratic bilinear form of $h$, and its
% integer coefficients.

\subsection{Bilinear construction}
\label{sec:construction}

For every $q\in[n]$, corresponding to the $q$-th Boolean coordinate of $\vb$, we introduce $2m+2$ independent variables
\begin{align*}
    \vx_q, \vy_q\in[0,1]^m\quad  \text{and} \quad s_q,s_q'\in[0,1].
\end{align*}
We refer to these variables collectively as \emph{block} $q$. The variables $s_q,s_q'$ serve to control block $q$. As we shall see, their average will encode the value associated with bit $q$ and determine whether the block is in the Boolean or fractional case. For each block $q$, we define the averages and differences
\begin{equation*}
 \vz_q\defeq \frac{\vx_q+\vy_q}{2},\quad
 \vd_q\defeq \vx_q-\vy_q,\quad
 \bar{s}_q\defeq \frac{s_q+s_q'}{2}.
\end{equation*}
Recalling that $w_{qr}$ denotes the integer coefficient in $h$, we choose
\begin{equation*}
 W=1+\max_{q\in[n]}\sum_{r\ne q}|w_{qr}|
\end{equation*}
and define $
 \tau=1/{(4mW)}.$
Finally, we define $\shiftedbit_q \defeq \bar{s}_q - \tau \sum_{j=1}^m d_{q,j}.$
The vector $\vshiftedbit = (\shiftedbit_1,\ldots,\shiftedbit_n)$ is the point at which $h$ is evaluated.
For every coordinate $q$, we have $|\sum_{j=1}^m d_{q,j}|\leq m$ and hence
\begin{equation}
 |\shiftedbit_q-\bar{s}_q|\leq\frac{1}{4W}.
 \label{eq:shift-bound}
\end{equation}
Since we extend $h$ from $\{0,1\}^n$ to $\mathbb R^n$ using the same polynomial, it remains well defined even if some coordinate $\shiftedbit_q$ lies outside $[0,1]$.

The reason why we shift the input bit $\shiftedbit_q$ from $\bar{s}_q$ is to incorporate the differences between $\vx_q$ and $\vy_q$ into the objective $h$. This allows the gradient of $h$ to appear directly in the derivatives of $\vx_q$ and $\vy_q$. The resulting terms can then be aligned with the derivatives with respect to the control variables $s_q$ and $s_q'$ through a common residual $R_q$, defined in~\eqref{eq:R}. At the same time, the choice of $\tau$ ensures that $\shiftedbit_q$ remains close to $\bar{s}_q$, which will allow us to preserve the local search information of $h$ when $s_q$ and $s_q'$ are Boolean variables.

We next construct a bilinear representation of the source quadratic objective $p$. The two copies $\vx_q$ and $\vy_q$ allow each squared term of $p$ to be replaced by a product between distinct variables, while preserving the source objective when $\vx_q$ = $\vy_q$. Specifically, we define
\begin{align}
 P(\vx_q,\vy_q)
&\defeq
\sum_i \mat{Q}_{ii}x_{q,i} y_{q, i}
+\frac12\sum_{i<j}\mat{Q}_{ij}(x_{q, i}+y_{q, i})(x_{q, j}+y_{q,j})
+\sum_i c_i(x_{q, i}+y_{q, i}).\label{eq:H0}
\end{align}
This guarantees that $P(\vz_q,\vz_q)=2p(\vz_q)$. However, $P$ alone does not guarantee that the two copies $\vx_q$ and $\vy_q$ agree. In particular, we have
$
\partial_{x_{q,i}}P(\vx_q,\vy_q)-\partial_{y_{q,i}}P(\vx_q,\vy_q)
=
-\mat{Q}_{ii}d_{q,i}.
$ When $\mat{Q}_{ii} > 0$, this difference in the derivatives drives $x_{q,i}$ and $y_{q,i}$ further apart. We therefore add an additional bilinear term so that the coefficient of $d_i$ in the resulting derivative difference is always positive.

Specifically, we define
\begin{align}
 \widehat{P}(\vx_q,\vy_q)
&\defeq
P(\vx_q,\vy_q)
-\sum_{i<j}(x_{q,i}-y_{q,i})(x_{q,j}-y_{q,j}).
 \label{eq:B}
\end{align}
As we show in \Cref{lem:derivatives}, this additional bilinear term introduces a positive $d_{q,i}$ term in the difference between the derivatives with respect to $x_{q,i}$ and $y_{q,i}$, which drives the two copies toward agreement. To conclude the construction, we let
$
\vx=(\vx_1,\ldots,\vx_n)\in[0,1]^{mn},
\vy=(\vy_1,\ldots,\vy_n)\in[0,1]^{mn},
$
and $\vs,\vs'\in[0,1]^n$ denote the control variables. We define the bilinear potential
\begin{equation}
 \Phi(\vx,\vy,\vs,\vs')
\defeq
h(\vshiftedbit)
+\sum_{q=1}^n
\left(
\frac{\tau^2}{4}\widehat{P}(\vx_q,\vy_q)
+\frac{\tau}{4}\bar{s}_q\sum_{j=1}^m d_{q,j}
-\frac18 s_q s_q'
\right),
\label{eq:Phi}
\end{equation}
which is to be minimized. The components in $\Phi$ are chosen so that its derivatives have a particular structure. The component $h(\vshiftedbit)$ incorporates the Boolean local search objective into the derivatives of $\vx_q$ and $\vy_q$, while $\widehat P(\vx_q,\vy_q)$ contributes the derivatives of the original objective together with the positive $d_{q,i}$ term discussed above. The third component combines the remaining terms from the first two components into a single scalar residual $R_q$ that appears in both the derivatives of $\vx_q,\vy_q$ and those of $s_q, s'_q$.  In the fractional case, the KKT conditions force this residual to vanish, which in turn drives $\vx_q=\vy_q$ (\Cref{lem:fractional}). In the Boolean case, the KKT conditions on this residual ensure that flipping any Boolean coordinate cannot decrease $h$ (\Cref{lem:boolean}). Finally, the last component $-s_qs_q'/8$ drives the two control variables together at every KKT point (\Cref{lem:control-agreement}). Accordingly, when $\Phi$ is converted into a network coordination game, $-\Phi$ serves as the payoff potential of the game (per~\Cref{lem:game-realization}).

\paragraph{Derivatives.} For each bit $q\in[n]$, we define the residual
\begin{equation}
 R_q\defeq \nabla h(\vshiftedbit)_q-\frac{\shiftedbit_q}{4}.
 \label{eq:R}
\end{equation}
We now calculate the derivatives of $\Phi$.
\begin{lemma}
\label{lem:derivatives}
For every $q\in[n]$ and $i\in[m]$,
\begin{alignat*}{2}
\partial_{x_{q,i}} \Phi
&=\frac{\tau^2}{4}\Biggl(F_i(\vz_q)+\left(1-\frac{\mat{Q}_{ii}}{2}\right)d_{q,i}\Biggr)-\tau R_q,
\qquad &\partial_{s_q} \Phi
&=\frac12R_q+\frac{s_q-s_q'}{16},\\
\partial_{y_{q,i}} \Phi
&=\frac{\tau^2}{4}\Biggl(F_i(\vz_q)-\left(1-\frac{\mat{Q}_{ii}}{2}\right)d_{q,i}\Biggr)+\tau R_q,
\qquad &\partial_{s_q'} \Phi
&=\frac12R_q-\frac{s_q-s_q'}{16}.
\end{alignat*}
\end{lemma}

We draw attention to the fact that the same scalar residual $R_q$ appears in all four derivatives. In particular, it enters with opposite signs for $\vx_q$ and $\vy_q$, and as the common component for $s_q$ and $s_q'$.

\begin{proof}[Proof of~\Cref{lem:derivatives}]
By the symmetry of $\mat{Q}$ and the definition of $P$ in \eqref{eq:H0}, for any bit $q \in [n]$ and any $\vx_q,\vy_q\in[0,1]^m$, we have
\[
 \partial_{x_{q,i}}P(\vx_q,\vy_q)
=
F_i\left(\frac{\vx_q+\vy_q}{2}\right)
-\frac{\mat{Q}_{ii}}{2}(x_{q,i}-y_{q,i}).
\]
Recalling that $d_{q,i}=x_{q,i}-y_{q,i}$, differentiating $\widehat P$ gives
\begin{equation*}
\begin{aligned}
 \partial_{x_{q,i}}\widehat P(\vx_q,\vy_q)
&=F_i\left(\frac{\vx_q+\vy_q}{2}\right)+\left(1-\frac{\mat{Q}_{ii}}{2}\right)d_{q,i}-\sum_{j=1}^m d_{q,j},\\
\partial_{y_{q,i}}\widehat P(\vx_q,\vy_q)
&=F_i\left(\frac{\vx_q+\vy_q}{2}\right)-\left(1-\frac{\mat{Q}_{ii}}{2}\right)d_{q,i}+\sum_{j=1}^m d_{q,j}.
\end{aligned}
\end{equation*}
Among the coordinates of $\vshiftedbit$, only $\shiftedbit_q$ depends on $x_{q,i}$, and $\partial_{x_{q,i}} \shiftedbit_q=-\tau.$ Therefore,
\begin{align*}
\partial_{x_{q,i}} \Phi
&=
-\tau \nabla h(\vshiftedbit)_q
+\frac{\tau^2}{4}
\Biggl(F_i(\vz_q)+\left(1-\frac{\mat{Q}_{ii}}{2}\right)d_{q,i}-\sum_{j=1}^m d_{q,j}\Biggr)
+\frac{\tau}{4}\bar{s}_q\\
&=
\frac{\tau^2}{4}
\Biggl(F_i(\vz_q)+\left(1-\frac{\mat{Q}_{ii}}{2}\right)d_{q,i}\Biggr)
-\tau
\left(
\nabla h(\vshiftedbit)_q
-\frac{\bar{s}_q-\tau \sum_{j=1}^m d_{q,j}}{4}
\right)\\
&=
\frac{\tau^2}{4}
\Biggl(F_i(\vz_q)+\left(1-\frac{\mat{Q}_{ii}}{2}\right)d_{q,i}\Biggr)
-\tau R_q,
\end{align*}
where the last equality follows from the fact that $\shiftedbit_q=\bar{s}_q-\tau \sum_{j=1}^m d_{q,j}$ and the definition of $R_q$. The expression for $\partial_{y_{q,i}} \Phi$ follows analogously from $\partial_{y_{q,i}}\shiftedbit_q=\tau$.

For the control variables, $\partial_{s_q}\shiftedbit_q=1/2$, and hence
\[
 \partial_{s_q} \Phi
=
\frac12\nabla h(\vshiftedbit)_q
+\frac{\tau \sum_{j=1}^m d_{q,j}}{8}
-\frac{s_q'}{8}.
\]
Using the definition of $R_q$ together with the fact that $\shiftedbit_q=\bar{s}_q-\tau \sum_{j=1}^m d_{q,j}$ and $\bar{s}_q=(s_q+s_q')/2$, we obtain the desired quantity. The expression for $\partial_{s_q'} \Phi$ follows symmetrically.
\end{proof}

\Cref{lem:derivatives} immediately implies the following property at a KKT point.

\begin{lemma}
\label{lem:control-agreement}
At every KKT point of $\Phi$, we have $s_q=s_q'=\bar{s}_q$ for every $q\in[n]$.
\end{lemma}
\begin{proof}
By \Cref{lem:derivatives}, we have
\[
 \partial_{s_q} \Phi -\partial_{s_q'} \Phi=\frac{s_q-s_q'}{8}.
\]
Suppose that $s_q>s_q'$. Then $s_q>0$ and $s_q'<1$, so by KKT conditions, it holds that $\partial_{s_q} \Phi \leq0$ and $
\partial_{s_q'} \Phi \geq0.$ Hence $\partial_{s_q} \Phi - \partial_{s_q'} \Phi \leq 0$, contradicting the fact that $s_q - s_q' > 0.$ The case where $s_q < s_q'$ is symmetric. Thus, we conclude that $s_q = s_q'$ at KKT points.
\end{proof}

\subsection{Fractional case}
\label{sec:fractional}

We proceed by analyzing what we call the fractional case, wherein $0<\bar{s}_q<1$ for some $q\in[n]$.

\begin{lemma}
\label{lem:fractional}
Let $(\vx,\vy,\vs,\vs')$ be a KKT point of $\Phi$. If $0<\bar{s}_q<1$ for some $q\in[n]$, then $\vx_q=\vy_q=\vz_q,$ and $\vz_q$ is an exact KKT point of $p$.
\end{lemma}
\begin{proof}
By \cref{lem:control-agreement}, if $0 < \bar{s}_q < 1$, then $s_q=s_q'\in(0,1)$. Hence, both control variables are interior, so their derivatives are zero. By \cref{lem:derivatives}, this implies $R_q=0$. Therefore, for every $i\in[m]$,
\begin{equation}
\begin{aligned}
 \partial_{x_{q,i}} \Phi
&=\frac{\tau^2}{4}
\Biggl(F_i(\vz_q)+\left(1-\frac{\mat{Q}_{ii}}{2}\right)d_{q,i}\Biggr),\\
\partial_{y_{q,i}} \Phi
&=\frac{\tau^2}{4}
\Biggl(F_i(\vz_q)-\left(1-\frac{\mat{Q}_{ii}}{2}\right)d_{q,i}\Biggr).
\end{aligned}
 \label{eq:fractional-derivatives}
\end{equation}
We consider the following two cases:
\begin{itemize}
    \item If $d_{q,i}>0$, then $x_{q,i}>0$ and $y_{q,i}<1$, so the KKT conditions give $\partial_{x_{q,i}} \Phi \leq 0$ and $\partial_{y_{q,i}} \Phi \geq 0$. Thus, $\partial_{x_{q,i}} \Phi-\partial_{y_{q,i}} \Phi \leq0$. On the other hand, since $|\mat{Q}_{ii}|\leq 1$ implies $1-\mat{Q}_{ii}/2\geq 1/2$, \eqref{eq:fractional-derivatives} gives
    \[
     \partial_{x_{q,i}} \Phi-\partial_{y_{q,i}} \Phi =\frac{\tau^2}{2}\left(1-\frac{\mat{Q}_{ii}}{2}\right)d_{q,i}>0,
    \]
    which is a contradiction.
    \item Similarly, if $d_{q,i}<0,$ then $x_{q,i} < 1$ and $y_{q,i} > 0.$ The KKT conditions imply $\partial_{x_{q,i}} \Phi \geq0$ and $\partial_{y_{q,i}} \Phi \leq0$. Hence, $\partial_{x_{q,i}} \Phi-\partial_{y_{q,i}} \Phi \geq0$, again contradicting~\eqref{eq:fractional-derivatives}.
\end{itemize}
We conclude that $d_{q,i}=0$ and thus $\vx_q = \vy_q = \vz_q$ for every
$i \in [m]$. Substituting $d_{q,i} = 0$ into \eqref{eq:fractional-derivatives} yields
$$\partial_{x_{q,i}} \Phi = \partial_{y_{q,i}} \Phi = \frac{\tau^2}{4} F_i(\vz_q).$$
If $x_{q,i} = y_{q,i} =0$, the KKT conditions give
$F_i(\vz_q)\geq0$; if $x_{q,i} = y_{q,i} =1$, they give $F_i(\vz_q)\leq 0$. If $x_{q,i} = y_{q,i} \in (0, 1),$ we have $F_i(\vz_q)=0$. We conclude that $\vz_q$ satisfies the KKT conditions of the original quadratic objective $p$.
\end{proof}

\subsection{Boolean case}
\label{sec:boolean-branch}

We next analyze what we refer to as the Boolean case, wherein $(\bar{s}_1,\ldots,\bar{s}_n) \in \{0,1\}^n$. The key idea in this case is to employ a different decoder, as we formalize below.

\begin{lemma}
\label{lem:boolean}
Let $(\vx,\vy,\vs,\vs')$ be a KKT point of $\Phi$. If $\bar{s}_q\in \{0,1\}$ for every $q\in[n]$, then
$\vb\defeq (\bar{s}_1,\ldots,\bar{s}_n)$ is a local minimum of $h$ under single-bit flips.
\end{lemma}
\begin{proof}
Since $\bar{s}_r=b_r$ for every $r\in[n]$, \eqref{eq:shift-bound} implies
$|\shiftedbit_r-b_r|\leq1/(4W)$. Therefore,
\begin{equation*}
 \left|\nabla h(\vshiftedbit)_q-\nabla h(\vb)_q\right|
\leq
\sum_{r\neq q}|w_{qr}|\cdot|\shiftedbit_r-b_r|
\leq\frac14.
\end{equation*}
Furthermore, since $b_q\in\{0,1\}$,
\[
 |\shiftedbit_q|
\leq
|b_q|+|\shiftedbit_q-b_q|
\leq
1+\frac{1}{4W}
\leq\frac54.
\]
By the definition of $R_q$ in \eqref{eq:R},
\begin{equation}
 \left|R_q-\nabla h(\vb)_q\right|
\leq
\frac14+\frac{5}{16}
<1.
 \label{eq:integer-margin}
\end{equation}
Since $h$ has integer coefficients and $\vb$ is a Boolean vector,
$\nabla h(\vb)_q\in\mathbb Z$. Hence, the bound in
\eqref{eq:integer-margin} ensures that $R_q$ cannot have the opposite sign from any nonzero value of $\nabla h(\vb)_q$.

By \cref{lem:control-agreement}, $s_q=s_q'=\bar{s}_q=b_q$. Hence,
\cref{lem:derivatives} gives $\partial_{s_q} \Phi=\partial_{s_q'} \Phi=R_q/2.$ If $b_q=0$, the KKT conditions imply $R_q\geq0$, and therefore \eqref{eq:integer-margin} implies $\nabla h(\vb)_q\geq0.$ If  $b_q=1$, we have $R_q\leq0$, hence $\nabla h(\vb)_q\leq0$. In either case,
$$
(1-2b_q)\nabla h(\vb)_q\geq0.
$$
By \eqref{eq:flip}, this is equivalent to $$h(\vb^{\oplus q})-h(\vb)\geq0.$$
Since this holds for every $q\in[n]$, we conclude that $\vb$ is a local minimum of $h$ under single-bit flips.
\end{proof}

\subsection{Putting everything together}
\label{sec:completion}

We are now ready to complete the reduction.

\begin{theorem}
\label{prop:exact-reduction}
$\QKKT\red\SFKKT$.
\end{theorem}

\begin{proof}
The output has polynomial encoding length. Given a rational KKT point of $\Phi$, the decoder checks the values $\bar{s}_q$. If some $\bar{s}_q$ lies in $(0,1)$, it returns $\vz_q$, which can be converted into a KKT point of $p$ by \cref{lem:fractional}. Otherwise, it returns $\mathcal D_p(\bar{\vs})$, which again can be converted into a KKT point of $p$ by \cref{lem:boolean-alternative,lem:boolean}.
\end{proof}

Combining with the $\CLS$-hardness for \QKKT~\citep{FGHS2025}, \Cref{thm:main} follows.

\section{Linear network congestion games}
\label{sec:network-congestion}

We next establish $\CLS$-completeness in network congestion games (introduced in~\Cref{sec:congestion-def}). We provide below a more detailed version of~\Cref{thm:congestion-hardness}.

\begin{theorem}
\label{thm:linear-network-congestion}
\NCGN is $\CLS$-complete, even on an unweighted directed acyclic network with linear latencies.
\end{theorem}

Membership in $\CLS$ is known from existing results~\citep{DP2011,Filos24:Ppad}, so we establish $\CLS$-hardness below. 

\begin{proof}[Proof of~\Cref{thm:linear-network-congestion}]
We reduce from \NCN. Let $\calG$ be a binary-action network coordination game with $N$ players and
payoff matrices $\mat{A}^{ij}$ (per~\Cref{sec:game-def}). We associate action $a\in\{0,1\}$ of player $i$ with the label $r(i,a)\defeq 2i-1+a$. We define symmetric weights by
\[
 w_{r(i,a),r(j,b)}\defeq 1-\mat{A}^{ij}_{ab}
 \quad\text{for }(i,j)\in E,
\]
and set all remaining weights to $0$. We first construct a game with nonnegative affine latencies, and then remove the constant terms.

\paragraph{The network congestion game.}
We adapt the reduction of~\citet{Ackermann08:Impact};
\Cref{fig:triangular-network} illustrates the construction. Each position $(r, c)$, with $1\leq c\leq r\leq 2N$, has separate horizontal and vertical entrances and exits. At a position $r>c$ belonging to different players, both entrances lead to a common edge of latency $w_{r c}k$, whose endpoint connects to both exits. At $(2i,2i-1)$, the horizontal entrance
connects only to the horizontal exit, and the vertical entrance only to the vertical exit, so that the two routes do not intersect. At $(r,r)$, the horizontal
entrance connects to the vertical exit.

We define $D \defeq 4 N ( 1 + W)$, where $W \defeq \sum_{r < r'} w_{r r'}$. We will also use the shorthand notation $W_r \defeq \sum_{r' = 1}^{2N} w_{r r'}$. Consecutive horizontal positions are connected by edges of constant latency $rD$ in row $r$, and consecutive vertical positions by edges of zero latency. Each player $i$ has a source node $s_i$ connected to the horizontal entrances of rows $2i-1,2i$ in column $1$ via \emph{source edges}, and a destination $t_i$ reached from columns $2i-1,2i$ in the bottom row via \emph{destination edges}.  Only player $i$ can use its source and
destination edges. For each label $r$, we give its source and destination edges constant latencies $f_r + W - W_r$ and $g_r$, respectively, where
\begin{equation}
 f_r\defeq D\left(4N^2-\frac{(r-1)(r-2)}2\right)\quad \text{and}\quad
 g_r\defeq D\left(4N^2-\frac{(r-1)(r+2)}2\right).
 \label{eq:network-tolls}
\end{equation}
These latencies are nonnegative. Now, let $P_r$ be the path that enters row $r$, follows it to $(r,r)$, and
then follows column $r$ to the destination. We call these paths \emph{canonical}. If $T \defeq 8 D N^2$, we have $f_r+Dr(r-1)+g_r=T$. That is, excluding the adjustment term $W - W_r$, the constant costs of the canonical paths are identical.

\begin{figure}[!t]
 \centering
 % Section 5: the adapted Ackermann--Roeglin--Voecking triangular grid.
% Include inside a figure using \input{triangular_grid.tikz}.
% Requires tikz and the arrows.meta library. All dimensions are in cm.
\begingroup
\definecolor{gridblue}{RGB}{0,92,153}
\definecolor{gridorange}{RGB}{184,83,0}
\definecolor{gridink}{RGB}{55,61,68}
\definecolor{gridgray}{RGB}{135,141,148}
\begin{tikzpicture}[
  x=1cm,y=1cm,
  font=\small,
  line cap=round,line join=round,
  >={Stealth[length=3.4pt,width=3pt]},
  grid edge/.style={draw=gridgray,line width=.45pt,->,
    shorten <=1.6mm,shorten >=1.6mm},
  boundary edge/.style={draw=gridgray,line width=.45pt,->},
  blue path/.style={draw=gridblue,line width=1.25pt},
  orange path/.style={draw=gridorange,line width=1.25pt,
    dash pattern=on 3pt off 1.6pt},
  terminal/.style={circle,draw=gridink,fill=white,line width=.65pt,
    inner sep=0pt,minimum size=3.1mm},
  shared cell/.style={rectangle,draw=gridink,fill=white,
    line width=.6pt,inner sep=0pt,minimum size=3.2mm},
  local edge/.style={draw=gridink,line width=.65pt,->},
  small label/.style={font=\footnotesize,text=gridink},
  panel label/.style={font=\small\bfseries,text=gridink,anchor=west}
]

% (a) The lower triangular array. Row indices increase downwards.
\begin{scope}
\colorlet{gridink}{black}
\colorlet{gridgray}{black}
\node[panel label] at (-.1,.75) {};
\foreach \r in {1,...,6} {
  \foreach \c in {1,...,\r} {
    \coordinate (g-\r-\c) at ({1.6+\c-1},{-.9*(\r-1)});
  }
}
\foreach \r in {2,...,6} {
  \pgfmathtruncatemacro{\lastcol}{\r-1}
  \foreach \c in {1,...,\lastcol} {
    \pgfmathtruncatemacro{\nextcol}{\c+1}
    \draw[grid edge] (g-\r-\c) -- (g-\r-\nextcol);
  }
}
\foreach \r in {1,...,5} {
  \pgfmathtruncatemacro{\nextrow}{\r+1}
  \foreach \c in {1,...,\r} {
    \draw[grid edge] (g-\r-\c) -- (g-\nextrow-\c);
  }
}

% Private sources and destinations: two labels for each original player.
\foreach \i in {1,2,3} {
  \pgfmathtruncatemacro{\oddlabel}{2*\i-1}
  \pgfmathtruncatemacro{\evenlabel}{2*\i}
  \node[terminal,label={[text=gridink]left:$s_\i$}]
    (s-\i) at (.15,{-1.8*(\i-1)-.45}) {};
  \node[terminal,label={[text=gridink]below:$t_\i$}]
    (t-\i) at ({2.1+2*(\i-1)},-5.55) {};
  \foreach \r in {\oddlabel,\evenlabel} {
    \draw[boundary edge,shorten >=1.6mm] (s-\i) -- (g-\r-1);
    \draw[boundary edge,shorten <=1.6mm] (g-6-\r) -- (t-\i);
  }
}

% P_2 (solid) and P_5 (dashed) belong to different players. Their unique
% shared resource is the edge represented by the square at (5,2).
\draw[blue path,->,shorten >=1.6mm] (s-1) -- (g-2-1);
\draw[grid edge,blue path] (g-2-1) -- (g-2-2);
\foreach \r in {2,...,5} {
  \pgfmathtruncatemacro{\nextrow}{\r+1}
  \draw[grid edge,blue path] (g-\r-2) -- (g-\nextrow-2);
}
\draw[blue path,->,shorten <=1.6mm] (g-6-2) -- (t-1);
\draw[orange path,->,shorten >=1.6mm] (s-3) -- (g-5-1);
\foreach \c in {1,...,4} {
  \pgfmathtruncatemacro{\nextcol}{\c+1}
  \draw[grid edge,orange path] (g-5-\c) -- (g-5-\nextcol);
}
\draw[grid edge,orange path] (g-5-5) -- (g-6-5);
\draw[orange path,->,shorten <=1.6mm] (g-6-5) -- (t-3);

% Representative latencies; the line key below states the general rule.
\path[font=\scriptsize,text=gridink]
  (g-2-1) -- node[midway,above=1pt,inner sep=1pt] {$2D$} (g-2-2)
  (g-3-2) -- node[midway,above=1pt,inner sep=1pt] {$3D$} (g-3-3)
  (g-4-3) -- node[midway,above=1pt,inner sep=1pt] {$4D$} (g-4-4)
  (g-5-3) -- node[midway,above=1pt,inner sep=1pt] {$5D$} (g-5-4)
  (g-6-4) -- node[midway,above=1pt,inner sep=1pt] {$6D$} (g-6-5)
  (g-1-1) -- node[midway,right=1pt,inner sep=1pt] {$0$} (g-2-1)
  (g-3-2) -- node[midway,left=1pt,inner sep=1pt] {$0$} (g-4-2)
  (g-3-3) -- node[midway,right=1pt,inner sep=1pt] {$0$} (g-4-3)
  (g-5-5) -- node[midway,left=1pt,inner sep=1pt] {$0$} (g-6-5);

% Expand the three kinds of cell symbol after drawing the connecting edges.
\foreach \r in {1,...,6} {
  \foreach \c in {1,...,\r} {
    \pgfmathtruncatemacro{\samplayer}{(mod(\r,2)==0)&&(\c==\r-1)}
    \begin{scope}[shift={(g-\r-\c)}]
      \ifnum\r=\c
        % The diagonal has only a horizontal entrance and a vertical exit.
        \draw[draw=gridink,line width=.65pt] (-.16,0) -- (0,0) -- (0,-.16);
        \ifnum\r=2
          \draw[blue path] (-.16,0) -- (0,0) -- (0,-.16);
        \fi
        \ifnum\r=5
          \draw[orange path] (-.16,0) -- (0,0) -- (0,-.16);
        \fi
      \else
        \ifnum\samplayer=1
          % A bridge denotes two disjoint edges, never a switching vertex.
          \draw[draw=gridink,line width=.65pt] (0,.16) -- (0,-.16);
          \ifnum\r=6
            \draw[orange path] (0,.16) -- (0,-.16);
          \fi
          \draw[draw=gridink,line width=.65pt,
            preaction={draw=white,line width=2.8pt}]
            (-.16,0) -- (-.065,0) arc[start angle=180,end angle=0,radius=.065]
            -- (.16,0);
          \ifnum\r=2
            \draw[blue path,preaction={draw=white,line width=3pt}]
              (-.16,0) -- (-.065,0) arc[start angle=180,end angle=0,radius=.065]
              -- (.16,0);
          \fi
        \else
          \node[shared cell] {};
        \fi
      \fi
    \end{scope}
  }
}
\node[shared cell,line width=1pt,fill=gridink!12] at (g-5-2) {};
\draw[draw=gridink,line width=.45pt]
  (2.44,-3.44) -- (2.25,-3.25);
\node[small label,inner sep=1pt] at (2.10,-3.13) {$(5,2)$};
\node[text=gridblue,anchor=west] at (2.7,-.9) {$P_2$};
\node[text=gridorange,anchor=west] at (5.7,-3.6) {$P_5$};

% Indices are offset from the boundary edges to keep every label legible.
\foreach \r in {1,...,6} {
  \node[font=\scriptsize,text=gridgray,anchor=south]
    at (1.29,{-.9*(\r-1)+.10}) {$\r$};
  \node[font=\scriptsize,text=gridgray,anchor=west]
    at ({1.6+\r-1+.16},-4.72) {$\r$};
}

% A compact line key also makes the selected paths distinguishable in print.
\draw[blue path] (.15,-6.48) -- (.7,-6.48);
\node[small label,anchor=west] at (.82,-6.48) {$P_2$ ($i=1$, $a=1$)};
\draw[orange path] (4.02,-6.48) -- (4.57,-6.48);
\node[small label,anchor=west] at (4.69,-6.48) {$P_5$ ($i=3$, $a=0$)};
\end{scope}

% (b) Local graph topology. A filled dot always denotes an actual vertex.
\node[panel label] at (8.15,.75) {};
\node[small label,text=black,anchor=west] at (8.15,.23)
  {Different players: shared edge at $(r,c)$};
\begin{scope}[shift={(10.9,-1.1)}]
  \coordinate (merge) at (-.48,0);
  \coordinate (split) at (.48,0);
  \draw[local edge] (-1.65,0) node[left] {$H_{\rm in}$} -- (merge);
  \draw[local edge] (-.48,.55) node[above] {$V_{\rm in}$} -- (merge);
  \draw[local edge,line width=1.3pt] (merge) --
    node[above=4pt,font=\small] {$w_{rc}k$} (split);
  \draw[local edge] (split) -- (1.65,0) node[right] {$H_{\rm out}$};
  \draw[local edge] (split) -- (.48,-.55) node[below] {$V_{\rm out}$};
  \fill[gridink] (merge) circle (1.5pt);
  \fill[gridink] (split) circle (1.5pt);
\end{scope}

\node[small label,text=black,anchor=west] at (8.15,-3.85)
  {Same player: disjoint routes at $(2i,2i-1)$};
\begin{scope}[shift={(10.9,-5.1)}]
  \draw[local edge] (0,.5) node[above] {$V_{\rm in}$}
    -- (0,-.5) node[below] {$V_{\rm out}$};
  \draw[local edge,preaction={draw=white,line width=4pt}]
    (-1.65,0) node[left] {$H_{\rm in}$} -- (-.14,0)
    arc[start angle=180,end angle=0,radius=.14]
    -- (1.65,0) node[right] {$H_{\rm out}$};
\end{scope}
\end{tikzpicture}
\endgroup
 \caption{Our construction, adapted from
 \citet{Ackermann08:Impact}, for $N=3$ players (corresponding to six action labels). Player $i$ enters through row $2i-1$ or $2i$ and leaves through column $2i-1$ or $2i$. The highlighted canonical paths $P_2$ and $P_5$ share the congestion edge at $(5,2)$. Squares abbreviate the shared edge shown on the right figure. Bridges at $(2i,2i-1)$ represent disjoint routes, so
 turning between the two labels of one player is impossible.}
 \label{fig:triangular-network}
\end{figure}
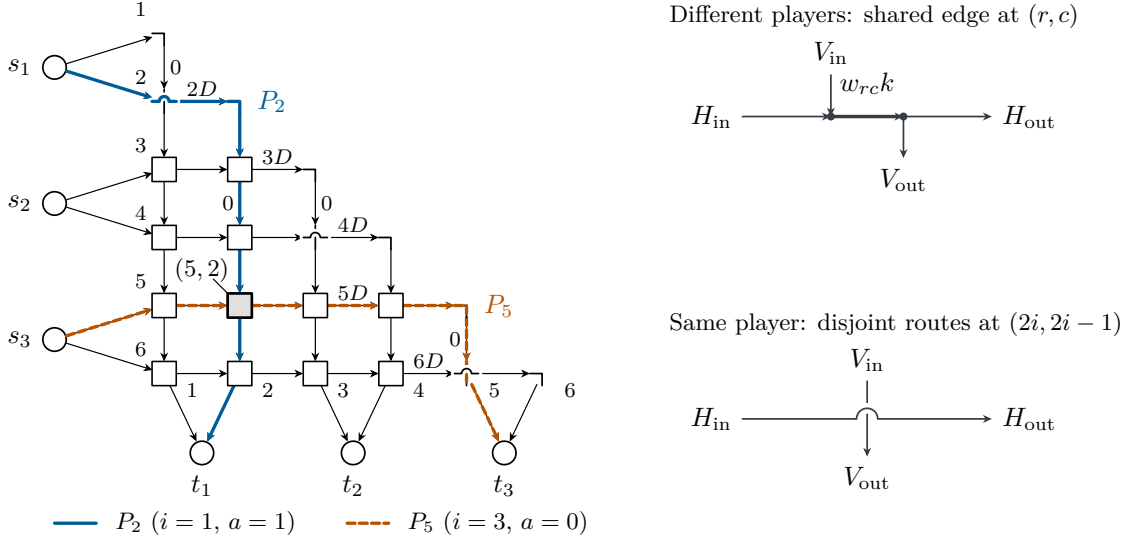

We will now argue that any noncanonical path is \emph{strictly dominated}. The basic idea is that when $D$ is large enough, the player always wants to choose a path that makes the horizontal crosses as early as possible, because horizontal edges become more costly at higher row indices. To formalize this, we consider a path that enters at row $r$ and exits at column $c$. Such a path must traverse exactly $c -1 $ horizontal edges. Moreover, the edge going from column $j$ to column $j+1$ must lie in a row with index at least $\max\{r, j+1\}$. As a result, the total cost from traversing these horizontal edges is at least $H(r,c) \defeq D\sum_{j=1}^{c-1}\max\{r,j+1\}$. The entrance and exit labels belong to the same player, so there are only three possibilities: $c \in \{r-1, r, r+1\}$. By construction of the latencies in~\eqref{eq:network-tolls}, we have $f_r+H(r,c)+g_c=T$ in all three cases. 

We now show that only a canonical path can attain the bound $H(r,c)$. If $c = r$, attaining the bound requires all horizontal edges to lie in row $r$, giving exactly the canonical path $P_r$. If $c=r-1$, it requires following row $r$ to column $r-1$, then turning down at $(r,r-1)$. If $c=r+1$, it requires following row $r$ to the diagonal, moving down to $(r+1,r)$, then turning right. In the latter two cases, the required turn is impossible because the row and column labels belong to the same player. We conclude that every noncanonical path has horizontal cost strictly greater than $H(r, c)$. Since horizontal costs are multiples of $D$, every noncanonical path has cost at least $T+ D+ W - W_r$.

Against arbitrary paths of the other players, the cost of $P_r$ is at most $T + W- W_r + N W_r \leq T + W - W_r + N W$, since each of its common edges is used
by at most $N$ players. Since $D > 4 N W$, every noncanonical path is strictly dominated by the canonical path with the same entrance.

When all players choose canonical paths, $P_r$ shares a common edge with each other chosen path and pays $W_r$ for using its common edges alone. As a result, at every pure profile $\va$,
\begin{equation}
 \mathrm{cost}_i\bigl(P_{r(1,a_1)},\ldots,P_{r(N,a_N)}\bigr)
 =T + W +\sum_{j\ne i}w_{r(i,a_i),r(j,a_j)}
 =T + W + \deg_G(i)-u_i(\va),
 \label{eq:network-payoffs}
\end{equation}
where $\deg_G(i)$ is $i$'s degree in the graph $G$ of the original network coordination game and $u_i(\va)$ is $i$'s utility in $\mathcal{G}$. Thus, the game restricted to canonical paths is strategically equivalent to $\calG$.

This completes the reduction for affine latencies. It remains to eliminate the constant terms in each latency function while preserving both the exclusion of noncanonical paths and the payoffs on canonical paths. 

First, we replace a source or destination edge of constant latency $z$ by latency
$zk$. If a player uses that edge, it will still pay exactly $z$. For every horizontal edge of
constant latency $z$, we add $K\defeq 16N^3$ auxiliary players whose source and destination are the tail and head of that edge, and replace its latency by
\[
 \ell_e(k)\defeq \frac{z}{K+1}k.
\]
Each auxiliary player has a unique path. No new edges are added. An original player using this edge now incurs cost $z(1+N' /(K+1))$, where $N' \leq N-1$ is the number of other original players using it. This is at least $z$, while the additional cost along $P_r$ is bounded by
\[
 \frac{N-1}{K+1}Dr(r-1)<\frac D4.
\]
Since the common-edge costs along $P_r$ are at most $N W < D/4$, every noncanonical path remains strictly dominated. On profiles corresponding to canonical paths, each horizontal edge is used by at most one original player, so~\eqref{eq:network-payoffs} still applies.

The constructed network has $O(N^2)$ vertices and edges, and $O(N^5)$ players. All coefficients have polynomial encoding length. Setting $x_i$ to be the probability of choosing $P_{2i}$ gives a Nash equilibrium $\vx$ of $\calG$. $\CLS$-hardness thus follows from~\Cref{thm:main}.
\end{proof}

% BEGIN sections/07_conclusion.tex
\section{Conclusion}
\label{sec:conclusion}

We showed that computing a Nash equilibrium is
$\CLS$-complete in network congestion games with linear latencies and also in network coordination games, even with two actions per player. As a corollary, we established $\CLS$-completeness for finding a KKT point of a bilinear polynomial over a box. In terms of future directions, we hope that our approach based on a two-way decoding argument---which is inspired by~\citet{FGHS2022,BC2026}---can find further uses to elucidate the complexity of other problems in $\TFNP$.

\appendix

\paragraph{Acknowledgments.} Ioannis Panageas and Jingming Yan were supported by NSF CCF-2454115. We used various versions of GPT Pro in the past months to explore proof strategies. In our current approach, GPT 6.0 Astra was used for minor derivations in~\Cref{sec:boolean-alternative,sec:reduction}, and for adapting the construction of~\citet{Ackermann08:Impact} starting from network coordination games in~\Cref{sec:network-congestion}.

\bibliography{main}

@inproceedings{BC2026,
  author       = {Martino Bernasconi and Matteo Castiglioni},
  title        = {The Complexity of Min-Max Optimization with Product Constraints},
  year         = {2026},
  booktitle       = {Symposium on Theory of Computing (STOC)}
}

@article{Goos24:Further,
  author       = {Mika G{\"{o}}{\"{o}}s and
                  Alexandros Hollender and
                  Siddhartha Jain and
                  Gilbert Maystre and
                  William Pires and
                  Robert Robere and
                  Ran Tao},
  title        = {Further Collapses in {TFNP}},
  journal      = {{SIAM} Journal on Computing},
  volume       = {53},
  number       = {3},
  pages        = {573--587},
  year         = {2024}
}

@inproceedings{Fearnley19:Unique,
  author       = {John Fearnley and
                  Spencer Gordon and
                  Ruta Mehta and
                  Rahul Savani},
  title        = {Unique End of Potential Line},
  booktitle    = {International Colloquium on Automata, Languages, and Programming (ICALP)},
  year         = {2019}
}

@article{Cai16:Zero,
  title={Zero-sum polymatrix games: A generalization of minmax},
  author={Cai, Yang and Candogan, Ozan and Daskalakis, Constantinos and Papadimitriou, Christos},
  journal={Mathematics of Operations Research},
  volume={41},
  number={2},
  pages={648--655},
  year={2016}
}

@article{Etessami10:Complexity,
  title={On the complexity of {N}ash equilibria and other fixed points},
  author={Etessami, Kousha and Yannakakis, Mihalis},
  journal={SIAM Journal on Computing},
  volume={39},
  number={6},
  pages={2531--2597},
  year={2010}
}

@article{Ishizuka21:Complexity,
  author       = {Takashi Ishizuka},
  title        = {On the complexity of finding a Caristi's fixed point},
  journal      = {Information Processing Letters},
  volume       = {170},
  pages        = {106119},
  year         = {2021}
}

@article{Fearnley17:CLS,
  author       = {John Fearnley and
                  Spencer Gordon and
                  Ruta Mehta and
                  Rahul Savani},
  title        = {{CLS:} New Problems and Completeness},
  journal      = {arXiv:1702.06017},
  year         = {2017}
}

@article{HY20:Hardness,
author = {Hub{\'a}{\v c}ek, Pavel and Yogev, Eylon},
title = {Hardness of Continuous Local Search: Query Complexity and Cryptographic Lower Bounds},
journal = {SIAM Journal on Computing},
volume = {49},
number = {6},
pages = {1128-1172},
year = {2020}
}

@inproceedings{Hollender25:Complexity,
  author       = {Alexandros Hollender and
                  Gilbert Maystre and
                  Sai Ganesh Nagarajan},
  title        = {The Complexity of Two-Team Polymatrix Games with Independent Adversaries},
  booktitle    = {International Conference on Learning Representations (ICLR)},
  year         = {2025}
}

@inproceedings{Daskalakis18:Converse,
  author       = {Constantinos Daskalakis and
                  Christos Tzamos and
                  Manolis Zampetakis},
  title        = {A converse to {B}anach's fixed point theorem and its {CLS}-completeness},
  booktitle    = {Symposium on Theory
                  of Computing (STOC)},
  year         = {2018}
}

@article{Anagnostides26:Algorithms,
  author       = {Ioannis Anagnostides and
                  Fivos Kalogiannis and
                  Ioannis Panageas and
                  Emmanouil{-}Vasileios Vlatakis{-}Gkaragkounis and
                  Stephen M. McAleer},
  title        = {Algorithms and complexity for computing {N}ash equilibria in adversarial
                  team games},
  journal      = {Games and Economic Behavior},
  volume       = {157},
  pages        = {138--152},
  year         = {2026}
}

@inproceedings{Tewolde24:Imperfect,
  author       = {Emanuel Tewolde and
                  Brian Hu Zhang and
                  Caspar Oesterheld and
                  Manolis Zampetakis and
                  Tuomas Sandholm and
                  Paul Goldberg and
                  Vincent Conitzer},
  title        = {Imperfect-Recall Games: Equilibrium Concepts and Their Complexity},
  booktitle    = {International Joint Conference on
                  Artificial Intelligence (IJCAI)},
  year         = {2024}
}

@inproceedings{Tewolde25:Computing,
  author       = {Emanuel Tewolde and
                  Brian Hu Zhang and
                  Caspar Oesterheld and
                  Tuomas Sandholm and
                  Vincent Conitzer},
  title        = {Computing Game Symmetries and Equilibria That Respect Them},
  booktitle    = {Conference on Artificial Intelligence (AAAI)},
  year         = {2025}
}

@inproceedings{Ghosh24:Complexity,
  author       = {Abheek Ghosh and
                  Alexandros Hollender},
  title        = {The Complexity of Symmetric Bimatrix Games with Common Payoffs},
  booktitle    = {Web and Internet Economics (WINE)},
  year         = {2024}
}

@inproceedings{CD2011,
  author    = {Yang Cai and Constantinos Daskalakis},
  title     = {On Minmax Theorems for Multiplayer Games},
  booktitle = {Symposium on Discrete Algorithms (SODA)},
  year      = {2011}
}

@inproceedings{DP2011,
  author    = {Constantinos Daskalakis and Christos Papadimitriou},
  title     = {Continuous Local Search},
  booktitle = {Symposium on Discrete Algorithms (SODA)},
  year      = {2011}
}

@inproceedings{Rubinstein16:Settling,
  author       = {Aviad Rubinstein},
  title        = {Settling the Complexity of Computing Approximate Two-Player {N}ash Equilibria},
  booktitle    = {Foundations of Computer Science (FOCS)},
  year         = {2016}
}

@article{Monderer96:Potential,
  title={Potential games},
  author={Monderer, Dov and Shapley, Lloyd S},
  journal={Games and economic behavior},
  volume={14},
  number={1},
  pages={124--143},
  year={1996}
}

@article{Johnson88:Easy,
  title={How easy is local search?},
  author={Johnson, David S and Papadimitriou, Christos H and Yannakakis, Mihalis},
  journal={Journal of computer and system sciences},
  volume={37},
  number={1},
  pages={79--100},
  year={1988}
}

@inproceedings{Fabrikant04:Complexity,
  title={The complexity of pure {N}ash equilibria},
  author={Fabrikant, Alex and Papadimitriou, Christos and Talwar, Kunal},
  booktitle={Symposium on Theory of Computing (STOC)},
  year={2004}
}

@article{Papadimitriou94:Complexity,
  title={On the complexity of the parity argument and other inefficient proofs of existence},
  author={Papadimitriou, Christos H},
  journal={Journal of Computer and system Sciences},
  volume={48},
  number={3},
  pages={498--532},
  year={1994}
}

@article{Rosenthal73:Class,
  title={A class of games possessing pure-strategy {N}ash equilibria},
  author={Rosenthal, Robert W},
  journal={International journal of game theory},
  volume={2},
  number={1},
  pages={65--67},
  year={1973}
}

@article{Daskalakis09:Complexity,
  title={The complexity of computing a {N}ash equilibrium},
  author={Daskalakis, Constantinos and Goldberg, Paul W and Papadimitriou, Christos H},
  journal={Communications of the ACM},
  volume={52},
  number={2},
  pages={89--97},
  year={2009}
}

@article{Chen09:Settling,
  title={Settling the complexity of computing two-player {N}ash equilibria},
  author={Chen, Xi and Deng, Xiaotie and Teng, Shang-Hua},
  journal={Journal of the ACM (JACM)},
  volume={56},
  number={3},
  pages={1--57},
  year={2009}
}

@inproceedings{BR2021,
  author       = {Yakov Babichenko and
                  Aviad Rubinstein},
  title        = {Settling the complexity of {N}ash equilibrium in congestion games},
  booktitle    = {Symposium on Theory of Computing (STOC)},
  year         = {2021}
}

@article{FGHS2022,
  author  = {John Fearnley and Paul W. Goldberg and Alexandros Hollender and Rahul Savani},
  title   = {The Complexity of Gradient Descent: {CLS} = {PPAD} $\cap$ {PLS}},
  journal = {Journal of the ACM},
  volume  = {70},
  number  = {1},
  pages   = {7:1--7:74},
  year    = {2022}
}

@article{FGHS2025,
  author       = {John Fearnley and
                  Paul W. Goldberg and
                  Alexandros Hollender and
                  Rahul Savani},
  title        = {The Complexity of Computing {KKT} Solutions of Quadratic Programs},
  journal      = {Journal of the {ACM}},
  volume       = {72},
  number       = {5},
  pages        = {31:1--31:48},
  year         = {2025}
}

@article{SY1991,
  author  = {Alejandro A. Sch{\"a}ffer and Mihalis Yannakakis},
  title   = {Simple Local Search Problems That Are Hard to Solve},
  journal = {SIAM Journal on Computing},
  volume  = {20},
  number  = {1},
  pages   = {56--87},
  year    = {1991}
}

@inproceedings{Skopalik08:Inapproximability,
  title={Inapproximability of pure {N}ash equilibria},
  author={Skopalik, Alexander and V{\"o}cking, Berthold},
  booktitle={Symposium on Theory of Computing (STOC)},
  year={2008}
}

@article{Myerson78:Refinements,
  title={Refinements of the {N}ash equilibrium concept},
  author={Myerson, Roger B},
  journal={International journal of game theory},
  volume={7},
  number={2},
  pages={73--80},
  year={1978}
}

@inproceedings{Tewolde23:Computational,
  author       = {Emanuel Tewolde and
                  Caspar Oesterheld and
                  Vincent Conitzer and
                  Paul W. Goldberg},
  title        = {The Computational Complexity of Single-Player Imperfect-Recall Games},
  booktitle    = {International Joint Conference on
                  Artificial Intelligence (IJCAI)},
  year         = {2023}
}

@inproceedings{Anagnostides26:Complexity,
      title={The Complexity of Equilibrium Refinements in Potential Games}, 
      author={Ioannis Anagnostides and Maria-Florina Balcan and Kiriaki Fragkia and Tuomas Sandholm and Emanuel Tewolde and Brian Hu Zhang},
      year={2026},
      booktitle={Conference on Economics and Computation (EC)}
}

@article{Selten75:Reexamination,
author = {Selten, R.},
title = {Reexamination of the perfectness concept for equilibrium points in extensive games},
year = {1975},
volume = {4},
number = {1},
journal = {International Journal of Game Theory},
pages = {25–55}
}

@article{Roughgarden15:Intrinsic,
  title={Intrinsic robustness of the price of anarchy},
  author={Roughgarden, Tim},
  journal={Journal of the ACM (JACM)},
  volume={62},
  number={5},
  pages={1--42},
  year={2015}
}

@inproceedings{Koutsoupias99:Worst,
  title={Worst-case equilibria},
  author={Koutsoupias, Elias and Papadimitriou, Christos},
  booktitle={Symposium on Theoretical Aspects of Computer Science (STACS)},
  year={1999}
}

@inproceedings{Caragiannis11:Efficient,
  title={Efficient computation of approximate pure {N}ash equilibria in congestion games},
  author={Caragiannis, Ioannis and Fanelli, Angelo and Gravin, Nick and Skopalik, Alexander},
  booktitle={Foundations of Computer Science (FOCS)},
  year={2011}
}

@article{Roughgarden02:Bad,
  title={How bad is selfish routing?},
  author={Roughgarden, Tim and Tardos, {\'E}va},
  journal={Journal of the ACM (JACM)},
  volume={49},
  number={2},
  pages={236--259},
  year={2002}
}

@inproceedings{Christodoulou05:Price,
  title={The price of anarchy of finite congestion games},
  author={Christodoulou, George and Koutsoupias, Elias},
  booktitle={Symposium on Theory of Computing (STOC)},
  year={2005}
}

@inproceedings{Filos24:Ppad,
  title={{PPAD}-membership for problems with exact rational solutions: a general approach via convex optimization},
  author={Filos-Ratsikas, Aris and Hansen, Kristoffer Arnsfelt and H{\o}gh, Kasper and Hollender, Alexandros},
  booktitle={Symposium on Theory of Computing (STOC)},
  year={2024}
}

@inproceedings{Daskalakis06:Game,
  title={The game world is flat: The complexity of {N}ash equilibria in succinct games},
  author={Daskalakis, Constantinos and Fabrikant, Alex and Papadimitriou, Christos H},
  booktitle={International Colloquium on Automata, Languages, and Programming (ICALP)},
  year={2006}
}

@article{Ackermann08:Impact,
  title={On the impact of combinatorial structure on congestion games},
  author={Ackermann, Heiner and R{\"o}glin, Heiko and V{\"o}cking, Berthold},
  journal={Journal of the ACM (JACM)},
  volume={55},
  number={6},
  pages={1--22},
  year={2008}
}

@inproceedings{Anagnostides25:Complexity,
  author       = {Ioannis Anagnostides and
                  Ioannis Panageas and
                  Tuomas Sandholm and
                  Jingming Yan},
  title        = {The Complexity of Symmetric Equilibria in Min-Max Optimization and
                  Team Zero-Sum Games},
  booktitle    = {Neural Information Processing Systems (NeurIPS)},
  year         = {2025}
}

@phdthesis{Nash50:Non,
  author     = {John Nash},
  title      = {Non-cooperative games},
  school     = {Princeton University},
  year       = {1950},
  department = {Department of Mathematics}
}
\end{document}